\documentclass[11pt,letter]{article}

\title{Sparsity Pattern Recovery in Bernoulli-Gaussian Signal Model}
\author{Subhojit Som, and Lee C. Potter
\vspace{0.2in}
\\
  Department of Electrical and Computer Engineering \\
  The Ohio State University, Columbus, OH, 43210, USA}

\usepackage{geometry}
\usepackage{color}
\usepackage{graphicx,psfrag}
\usepackage{amsfonts}
\usepackage{amssymb}
\usepackage{amsthm}
\usepackage{amsmath}
\usepackage{epsfig}
\usepackage{cite}

\theoremstyle{plain}  \newtheorem{theorem}{Theorem}
\theoremstyle{plain}  
\theoremstyle{plain}  \newtheorem{proposition}{Proposition}
\theoremstyle{plain}  \newtheorem{corollary}{Corollary}[theorem]
\theoremstyle{plain}

\newcommand{\ve}[1]{\ensuremath{\boldsymbol{#1}}}
\renewcommand{\sp}[1]{\ensuremath{\mathcal{#1}}}
\def\delRIP{\varepsilon}

\def\Ph{\ve{\Phi}}
\def\Phs{\ve{\Phi}(S)}
\def\Phsi{\ve{\Phi}(S)^{-1}}

\begin{document}

\date{}


\maketitle

\begin{abstract}
In compressive sensing, sparse signals are recovered from underdetermined noisy linear
observations. One of the interesting problems which attracted a lot of attention in recent times is the 
support recovery or sparsity pattern recovery problem. The aim is to identify the non-zero elements in 
the original sparse signal. In this article we consider the sparsity pattern recovery problem under a 
probabilistic signal model where
the sparse support follows a Bernoulli distribution and the signal restricted to this support
follows a Gaussian distribution. We show that the energy in the original signal 
restricted to the missed support of the MAP estimate is bounded above and this bound is of the 
order of 
energy in the projection of the noise signal to the subspace spanned by the active coefficients.
We also derive sufficient conditions for no misdetection and no false alarm in support recovery.
\end{abstract}



\section{Introduction} \label{sec:BMAP:intro}
We consider the linear observation model 
\begin{eqnarray}
\ve{y} = \ve{A} \ve{x} + \ve{e},  \label{eqn:sys}
\end{eqnarray}
where $\ve{x} \in \mathbb{R}^N$ is the signal vector, $\ve{e} \in \mathbb{R}^M$ is the noise vector, 
$\ve{A} \in \mathbb{R}^{M \times N}$ is the measurement matrix, and $M \ll N$. 
In spite of this being an ill-posed problem, various algorithms have been proposed for estimation of the unknown 
signal $\ve{x}$ and performance guarantees have been proven for them subject to sparsity  
of the signal $\ve{x}$ and some coherence constraints on the measurement matrix $\ve{A}$. 
This technique is known as compressive sensing or compressive 
sampling~\cite{ref:compressedsensing06, ref:CandesTaoNoiseless06, 
ref:StableSparsedonoho06, ref:CandesTaoNoisy06, ref:TroppRelax06} and it has received a lot of 
attention in recent past among researchers.

In this article we consider the problem of sparse support recovery, also known as sparsity pattern 
recovery, where the aim is to identify the indices of the non-zero elements of $\ve{x}$.  The main 
contribution of this article is non-asymptotic analysis of support recovery in terms of quality of the
recovered support set. We analyze how much energy of the true signal remains in the missed coefficients
under Bernoulli-Gaussian signal prior assumption. We also derive a sufficient condition for perfect 
support recovery under this signal model.

In section~\ref{sec:BMAP:related_work} we discuss the most relevant prior work related to this article, in 
section~\ref{sec:BMAP:contrib} we briefly describe the contribution of this article. 
In section~\ref{sec:BMAP:signal_model} we describe the probabilistic signal model for the variable~$\ve{x}$,
in section~\ref{sec:BMAP:coherence} the coherence property of the measurement matrix~$\ve{A}$ is defined. 
Section~\ref{sec:BMAP:MAP} outlines the support recovery problem and defines the estimator for support set.
Two theorems regarding the energy bound on the missed support and sufficiency condition for perfect support 
recovery are stated in sections~\ref{sec:BMAP:energy} and~\ref{sec:BMAP:perf_recov} respectively. The proofs are 
given in section~\ref{sec:BMAP:proofs} and the results are discussed in section~\ref{sec:BMAP:disc}.

%

\subsection{Related Work} \label{sec:BMAP:related_work}
Significant amount of work has been done in recent times on signal recovery in compressive sensing. The $\ell_2$-norm
of error in estimating the signal $\ve{x}$ is the most popular performance 
metric~\cite{ref:CandesTaoNoisy06,ref:StableSparsedonoho06,ref:CoSaMP}, 
but in the noisy setting stability of the solution and boundedness of this performance metric do not give any direct 
guarantee about support recovery. Here we briefly describe the sparsity pattern recovery results most relevant to
our work. Donoho {\em et al.} showed in their work~\cite{ref:StableSparsedonoho06}
that $\ell_1$-constrained quadratic program with exaggerated noise level guarantees partial support recovery. They 
also derived the upper bound on the number of non-zero elements in the signal vector in terms of mutual coherence 
of the measurement matrix and minimum absolute value of the non-zero elements in the true signal for perfect support
recovery using an orthogonal greedy algorithm. Cand\`{e}s {\em et al.} showed in~\cite{ref:CandesMSL1} that if the 
measurement matrix satisfies certain coherence properties and the signs of the non-zero elements of the signal 
are equally likely to be positive and negative then $\ell_1$-regularized least squares solution recovers the signed
support perfectly with very high probability when the regularization constant is chosen appropriately and the
minimum absolute value of the non-zero elements of the signal is above certain threshold. Recovery of 
signed support means
the support sets of true signal and the estimate are identical and the non-zero elements in the true signal and the 
estimate have the same signs. Zhao {\em et al.} showed in~\cite{ref:MSelLASSO}
that the irrepresentable condition is almost necessary and sufficient for LASSO to select the true model both in 
the classical fixed $N$ setting and in the large $N$ setting as the observation size $M$ gets large. At some special 
scenarios this irrepresentable condition coincides with the coherence condition used in the work of Donoho {\em et al.}
A similar condition is used by Meinshausen {\em et al.} in~\cite{ref:MeinModelSel} to prove a model selection 
consistency result for Gaussian graphical model selection using the LASSO.
Using replica method Guo~{\em et al.} showed~\cite{ref:GuoSL2009} that the posterior distribution of estimating a 
single coefficient becomes asysmptotically decoupled from estimation of other coefficients. Detecting a single
coefficient is analogous to detecting this input coefficient with all other coefficients suppressed, but based
on a noisier observation. They derived 
the maximum probability of making an error in detecting a single coefficient and the corresponding MMSE under the 
high SNR and large system limits. 
Rangan~{\em et al.}~\cite{ref:RanganReplica09} use the same replica claim framework to obtain the mean squared error 
in estimation of the variable $\ve{x}$ under the large system limits for linear, LASSO, and zero-norm regularized
estimators.

There is another class of papers where the minimum number of observations~$M$ needed for perfect support recovery or 
partial support recovery expressed as a fraction of the true support size is 
investigated~\cite{ref:ReevesSB, ref:FletcherNS, ref:Wainwright1, ref:Wainwright2, ref:Akcakaya}. 
In these articles it is assumed that elements of the measurement matrix are i.i.d. Gaussian.
Necessary and sufficient conditions for exhaustive search based decoders and $\ell_1$-constrained least squares
are derived in these articles.

\subsection{Contributions} \label{sec:BMAP:contrib}
Our results are non-asymptotic with fixed model dimensions. Except~\cite{ref:StableSparsedonoho06} 
and~\cite{ref:CandesMSL1} other support recovery results for linear observation model, discussed 
in section~\ref{sec:BMAP:related_work}, are 
aysmptotic analyses. Our first result is about partial support recovery. 
We characterize any support set in terms of the energy in the true signal restricted to this support set.
More specifically, we explore the 
relationship between energy in the missed support and the noise energy under the probabilistic model where the
signal prior is known. 
Most earlier partial support recovery results characterize the fraction of the support recovered
{\em i.e.,} they do not not distinguish between
missing the coefficient with the highest absolute value and the lowest absolute value
but our performance metric captures that.
To the best of our knowledge the only exception is the work by
Akcakaya {\em et al.}~\cite{ref:Akcakaya}. 
They investigated the number of measurements needed for partial support 
recovery in terms of fraction of total energy in the true signal restricted to the recovered support. But 
their analysis is asymptotic whereas we have considered fixed model dimensions.
Our second result is about sufficient conditions for guaranteeing no missed coefficient and no false 
detection for this Bernoulli-Gaussian signal model when the absolute value of any active coefficient 
is bounded below with a very high probability.


\section{Problem Statement} \label{sec:BMAP:prob_state}
\subsection{Signal Model}  \label{sec:BMAP:signal_model}
We consider a probabilistic signal model for the sparse signal $\ve{x}  \in \mathbb{R}^N$. 
Let $S$ be a set whose entries are drawn from the 
set $I=\{1, 2, \ldots, N \}$ in such a way that each entry of 
$I$ is in the set $S$ with probability $p \ll1$ and their inclusion in $S$ is independent of each other. 
Thus the probability that the cardinality of the support set $S$ equals $K$ is given by
$\mathbb{P}[|S|=K] = {N \choose K}p^K(1-p)^{N-K}$. To enforce sparsity we also assume that $p < \frac{1}{2}$.
Each element of $\ve{x}$ is identically zero if the corresponding index is not in the set $S$,
otherwise the element is Gaussian with mean $\mu_1$ and non-zero variance $\sigma_1^2$. The mean $\mu_1$ can be zero or 
non-zero.  Elements of $\ve{x}$ are distributed independently given the support set.
If $\ve{x}_S$ denotes the vector consisting of the elements of \ve{x} whose indices are in the set $S$, 
then the vector $\ve{x}_S$ follows i.i.d. Gaussian distribution {\em i.e.,} 
$\ve{x}_S \sim \mathcal{N}(\mu_1 \ve{1}_{|S|}, \sigma_1^2 \ve{I}_{|S|})$
\footnote{The vector of ones of size $|S| \times 1$ is denoted by $\ve{1}_{|S|}$. Similarly the vector of ones
of size $|S_1| \times 1$ is denoted by $\ve{1}_{|S_1|}$. It is also denoted by $\ve{1}_1$ when there is no ambiguity.
The notations $\ve{1}_{|S_{01}|}$ and $\ve{1}_{01}$ are used interchangeably. The same applies to the subscripts used
for the identity matrix $\ve{I}$.}.
Thus $S$ is the support set of the signal vector $\ve{x}$ with expected cardinality 
$\mathbb{E}[|S|] =Np \ll N$ and $\ve{x}$ is sparse with high probability.  
This Bernoulli-Gaussian model is quite popular in 
literature for a long time~\cite{ref:George1, ref:George2, ref:Haystack} 
for modeling sparse vectors in diverse application areas and is also becoming increasingly popular in the 
compressive sensing research~\cite{ref:FBMP2008, ref:BPCS, ref:GuoSL2009}.

\subsection{Coherence of Measurement Matrix}  \label{sec:BMAP:coherence}
Several conditions have been proposed which characterize coherence properties of the measurement 
matrix $\ve{A}$ and are used for deriving any performance
guarantee for compressive sensing algorithms. Measurement matrix with entries drawn from i.i.d. Gaussian or
Bernoulli distributions, and partial Fourier matrix are known to satisfy these properties.
In~\cite{ref:StableSparsedonoho06}, it is shown that if the mutual coherence {\em i.e.,} the 
magnitude of the maximum entry of the Gram matrix $m(\ve{A}) = \max_{i,j:i \neq j} |(\ve{A}^T \ve{A})_{i,j}|$ is
small then robust signal and support recovery is possible for sparse signals. Another condition known as 
restricted isometry \mbox{property (RIP)} is proposed in~\cite{ref:CandesTaoNoisy06}.
Here we assume that the measurement matrix $\ve{A}$ satisfies RIP with $(4Np,\delRIP)$,
{\em i.e.,} for any sparse vector $\ve{x}$ with cardinality of support set $\leq 4Np$, 
\begin{eqnarray}
(1-\delRIP) \| \ve{x} \| _2^2 \leq  \| \ve{A} \ve{x} \| _2^2 \leq (1+\delRIP) \| \ve{x} \| _2^2. \label{eqn:RIP}
\end{eqnarray}
Though determination of RIP of a given matrix is a NP-hard problem, it can be shown~\cite{ref:BaraniukRIP} that 
random matrices satisfy RIP properties with overwhelming probability. In contrary mutual coherence is a verifiable 
condition but it gives much weaker performance guarantee than RIP.

We note here that the \mbox{constant 4} in the definition of RIP of $\ve{A}$ is arbitrary and a matter of convenience. 
In this article we also assume that $\delRIP \leq \frac{1}{3}$ in order to obtain simple expressions in our results.
Leaving $\delRIP$~as a parameter makes the results difficult to interpret.
We can always choose any other constant instead of 4 in definition of RIP for the measurement matrix 
and a different upper bound on $\delRIP$-value. This will lead to different values of the constants appearing in our results.

\subsection{Support Recovery} \label{sec:BMAP:MAP}
In this article we consider the problem of support recovery {\em i.e.,} identifying the indices corresponding to the
Gaussian with $\sigma_1^2$ variance. 
Assuming additive white Gaussian noise with variance $\sigma_e^2$, {\em i.e.,} $\ve{e} \sim \mathcal{N}(\ve{0}, \sigma_e^2 \ve{I}_M)$,
\begin{eqnarray}
\left. \ve{y}   \right| S \sim \mathcal{N}\left( \mu_1 \ve{A}_S \ve{1}_{|S|}, \Phs \right),  \label{eqn:model_distr}
\end{eqnarray}
where $\Phs$ is given by, 
\begin{eqnarray}
\Phs = \sigma_1^2 \ve{A}_S \ve{A}_S^T + \sigma_e^2 \ve{I}_M.  \label{eqn:Phs}
\end{eqnarray}
The maximum \mbox{{\em a posteriori} (MAP)} estimate of the support set is given by,
\begin{eqnarray}
\widehat{S}_{\mathrm{MAP}} &=& \arg \max_S p(S|\ve{y}) = \arg \max_S p(\ve{y}|S) p(S)\nonumber\\
&=&  \arg \max_S \int_{\ve{x}} p(\ve{y}|\ve{x},S) p(\ve{x}|S) d\ve{x} \cdot p(S) \nonumber\\
&=& \arg \min_{S} \frac{1}{2} \ln \det (\Phs) + \frac{1}{2} (\ve{y}-\mu_1 \ve{A}_S\ve{1}_{|S|})^T\Phsi (\ve{y}-\mu_1 \ve{A}_S\ve{1}_{|S|}) +  |S| \ln \frac{1-p}{p}. \label{eqn:MAP}
\end{eqnarray}

We have adopted a probabilistic model for the number of active elements, the signal and the noise. Though the 
number of non-zero elements in $\ve{x}$ has mean $Np \ll N$, it can be as large as $N$ with very small 
but non-zero probability. Similarly signal and noise energy can be arbitrarily large with vanishingly small but
non-zero probability. Nevertheless the quantities like cardinality and energy are bounded with overwhelmingly high
probability. 
Keeping this in mind we study the suboptimal estimator which minimizes the MAP cost function subject to the 
constraint $|S| \leq 2Np$: 
\begin{eqnarray}
\widehat{S} = \arg \min_{S: |S| \leq 2Np} \frac{1}{2} \ln \det (\Phs) + \frac{1}{2} (\ve{y}-\mu_1 \ve{A}_S\ve{1}_{|S|})^T\Phsi (\ve{y}-\mu_1 \ve{A}_S\ve{1}_{|S|}) +  |S| \ln \frac{1-p}{p}. \label{eqn:MAP2}
\end{eqnarray}
We define the event $E$ to be the cardinality of the true support being less than or equal to $2Np$. As we see later 
that event $E$ holds with high probability and the estimator defined in~(\ref{eqn:MAP2}) satisfies 
certain performance criteria if event $E$ holds. 
Here we emphasize that instead of $2NP$ we can use $LNp$ for any other $L>1$ in the definition of the event as $E$. Similarly 
we can use any other constraint $|S| \leq QNp$ in the definition of $\widehat{S}$ in~(\ref{eqn:MAP2}), where $Q>1$. The 
choice of $L=Q=2$ is arbitrary in the definition of the event $E$ and the definition of $\widehat{S}$
but related to the constant used in the definition of RIP satisfied by the measurement 
\mbox{matrix \ve{A}}. They are chosen in such a way that $L+Q \leq n$, when $\ve{A}$ satisfies 
RIP with $(nNp,\delRIP)$. As mentioned earlier we have arbitrarily chosen $n=4$.

\subsection{Energy in Missed Coefficients} \label{sec:BMAP:energy}
Our first theorem, as stated below, shows that the total energy in the missed coefficients is of 
the order of the average energy in the projection of noise to the subspace spanned by the active columns of 
the \mbox{$\ve{A}$ matrix}. Here we make no assumption about the mean of the Gaussian distribution $\mu_1$.
\begin{theorem}[Energy Bound on Missed Coefficients]\label{th:BMAP1}
For the signal and observation models under consideration, the $\ell_2$-norm of the signal restricted to the
index set of missed coefficients is upper bounded by $K_1 \sqrt{Np} \sigma_e$ with probability exceeding 
$(1-e^{-Np(2 \ln2 -1)})(1-3 e^{- Np(\beta-1- \ln \beta)})$, where 
$K_1=2 \left( \sqrt{7 \beta + C} + \sqrt{ \beta} \right)$, 
$C=\ln \left(1+ \frac{4\sigma_1^2}{3\sigma_e^2} \right) + 2 \ln \frac{1-p}{p}$ and $\beta>1$.
\end{theorem}
Different values of the parameter $\beta$ give different values of the constant $K_1$ and also the probability
with which the energy in the missed coefficients is bounded by $K_1^2 Np \sigma_e^2$. Both $K_1$ and the minimum probability 
are increasing function of $\beta$. This is natural since as we increase the bound, {\em i.e.,} make it loose, the 
probability with which it is satisfied also increases. We also see that the constant $C$ is dependent on $p$ and the ratio
$\sigma_1^2/\sigma_e^2$. Thus the constant $K_1$ increases as the signal model is known to be more sparse. The dependence 
of $K_1$ on $\sigma_1^2/\sigma_e^2$ is a bit counterintuitive. As we discuss in section~\ref{sec:BMAP:disc}, this bound 
becomes loose at high SNR. At very high value of this ratio there is no missed coefficient with a very high probability.

\subsection{Perfect Support Recovery}  \label{sec:BMAP:perf_recov}
It is hard to recover the support set perfectly for the zero-mean signal model since a significant number of
coefficients are close to zero. Hence they are almost impossible to detect in the presence of noise. 
If the signal mean is 
high enough to ensure that all the coefficients are well above the noise level then all of them are detected
with a high probability. But even then ensuring that no false alarm happens is tough. It requires even higher
value of the mean. The following theorem states these results.
\begin{theorem}[Sufficient Condition for Perfect Support Recovery]\label{th:BMAP2}
For the signal and observation models under consideration, all active coefficients are selected
{\em i.e.,} there is no missed coefficient with probability exceeding 
$(1-e^{-Np(2 \ln2 -1)})(1-3 e^{- Np(\beta-1- \ln \beta)} - e^{- \frac{(\bar{\beta}-1 - \ln \bar{\beta})}{2}})$
if $|\mu_1| > K_2 \sigma_1 + K_1 \sqrt{Np} \sigma_e$ where $K_2=\sqrt{\bar{\beta}}$, 
and $\beta, \bar{\beta} > 1$.  $K_1$ and $C$~are as defined in theorem~\ref{th:BMAP1}.
Perfect support recovery happens with the same probability if $|\mu_1| > K_3 \sigma_1 + K_4 \sqrt{Np} \sigma_e$, where
$K_3=\max\{K_2,6 \sqrt{2\beta Np} \}$ and $K_4= \max\{K_1, 3 \left(\frac{1}{2} + \sqrt{3}\right) \sqrt{2\beta}\}$. 
\end{theorem}
Here the condition  $|\mu_1| > K_2 \sigma_1 + K_1 \sqrt{Np} \sigma_e$ is needed for probabilistic guarantee for no
misdetection. This condition implies that if the distribution of $\ve{x}_S$ is such that with very high probability 
absolute values of all the elements are above the noise level in the subspace spanned by the active columns of the 
measurement matrix then with very high probability there is no active coefficient excluded from~$\widehat{S}$.
In addition to this condition, we also need $|\mu_1| >  6 \sqrt{2\beta Np} \sigma_1 + 3 \left(\frac{1}{2} + \sqrt{3}\right) \sqrt{2\beta} \sqrt{Np} \sigma_e$
for guarantee on no false alarm.

\section{Proofs} \label{sec:BMAP:proofs}
\subsection{Some Propositions}\label{sec:BMAP_proofs:prop}
Before proceeding further we provide the following propositions. The first proposition is a consequence of RIP. It shows 
near orthonormality of the columns of $\ve{A}$ matrix {\em i.e.,}  
the column spaces of any two submatrices $\ve{A}_i$ and $\ve{A}_j$ 
of the matrix $\ve{A}$ are almost orthogonal to each other if $S_i \cap S_j =\emptyset$ and $|S_i|+|S_j|\leq 4Np$.
\begin{proposition}\label{proposition:bp}
If $S_i \subset \{1, 2, \ldots,  N \} $, $S_j \subset \{1, 2, \ldots,  N \} $, $S_i \cap S_j = \emptyset$, $\ve{A}$ 
satisfies RIP with $(4Np,\delRIP)$ and $|S_i|+|S_j| \leq 4Np$, then 
the vector induced norm $ \| \ve{A}_i^T \ve{A}_j \| _2 \leq \delRIP$.
\end{proposition}
\begin{proof}
This proof is due to~\cite{ref:CoSaMP}. 
Let $S=S_i \cup S_j$. Note that $\ve{A}_i^T\ve{A}_j$ is a submatrix of $\ve{A}_S^T\ve{A}_S-\ve{I}_{|S|}$.
Since the induced norm of a submatrix never exceeds the norm of the matrix, 
\begin{eqnarray}
\| \ve{A}_i^T \ve{A}_j \|_2 \leq \|\ve{A}_S^T\ve{A}_S-\ve{I}_{|S|} \|_2 \leq \max \{(1+\delRIP)-1, 1-(1-\delRIP) \} =\delRIP, 
\end{eqnarray}
since the singular values of the matrix $\ve{A}_S^T\ve{A}_S$ lie between $1-\delRIP$ and $1+\delRIP$. 
\end{proof}
\begin{proposition}\label{prop:bp1}
Let $\ve{A}_i=\ve{U}_i \ve{\Sigma}_i \ve{V}_i^T$ be the Singular Value Decomposition (SVD) of $\ve{A}_i$. Let
$\bar{\ve{U}}_i$ be the submatrix formed by taking the first $|S_i|$ columns of $\ve{U}_i$ and
$\underbar{\ve{U}}_i$ be the submatrix formed by taking the rest $M-|S_i|$ columns of $\ve{U}_i$.
If $\ve{x} \in \mathbb{R}^{|S_j|}$, then $\| \bar{\ve{U}}_i^T \ve{A}_j \ve{x}\|_2 \leq \frac{\delRIP}{\sqrt{1-\delRIP}} \|\ve{x} \|_2$ and  $\| \underbar{\ve{U}}_i^T \ve{A}_j \ve{x}\|_2 \geq \sqrt{\frac{1- 2 \delRIP}{1-\delRIP}} \|\ve{x} \|_2$. Also, if $\ve{v} \in \mathbb{R}^{|S_j|}$, then $\| \underbar{\ve{U}}_i^T \bar{\ve{U}}_j \ve{v}\|_2 \geq \sqrt{\frac{1- 2 \delRIP}{1-\delRIP^2}} \|\ve{v} \|_2$ where $\bar{\ve{U}}_j$ is defined similar to $\bar{\ve{U}}_i$.
\end{proposition}
\begin{proof}
From proposition~\ref{proposition:bp}, $\| \ve{A}_i^T \ve{A}_j \ve{x} \|_2 \leq \delRIP \|\ve{x}\|_2$ and
$\| \ve{A}_i^T \ve{A}_j \ve{x} \|_2 = \|\ve{V}_i \ve{\Sigma}_i^T \ve{U}_i^T \ve{A}_j \ve{x} \|_2 =   \|\bar{\ve{\Sigma}}_i^T \bar{\ve{U}}_i^T \ve{A}_j \ve{x} \|_2 $
where $\bar{\ve{\Sigma}}_i$ is the upper left $|S_i| \times |S_i|$ diagonal submatrix of $\ve{\Sigma}_i$.
Thus $\|\bar{\ve{\Sigma}}_i^T \bar{\ve{U}}_i^T \ve{A}_j \ve{x} \|_2 \leq \delRIP \|\ve{x}\|_2$.  Since 
elements on $\mathrm{diag}({\bar{\ve{\Sigma}}_i}) \geq \sqrt{1-\delRIP}$, we conclude that 
$\| \bar{\ve{U}}_i^T \ve{A}_j \ve{x}\|_2 \leq \frac{\delRIP}{\sqrt{1-\delRIP}} \|\ve{x} \|_2$. 
Now $\| \ve{A}_j \ve{x}\|_2^2 \geq (1-\delRIP) \| \ve{x}\|_2^2$. Thus 
$\|\underbar{\ve{U}}_i^T \ve{A}_j \ve{x}\|_2 \geq \sqrt{(1-\delRIP)-\frac{\delRIP^2}{1-\delRIP}} \|\ve{x}\|_2 = \sqrt{\frac{1-2\delRIP}{1-\delRIP}} \|\ve{x}\|_2$.
Now we can rewrite this as  $\|\underbar{\ve{U}}_i^T \bar{\ve{U}}_j \bar{\ve{\Sigma}}_j \ve{V}_j^T \ve{x}\|_2 \geq \sqrt{\frac{1-2\delRIP}{1-\delRIP}} \|\ve{x}\|_2$. Taking $\ve{v}=\bar{\ve{\Sigma}}_j \ve{V}_j^T \ve{x}$, we see that
 $\|\underbar{\ve{U}}_i^T \bar{\ve{U}}_j \ve{v}\|_2 \geq \sqrt{\frac{1-2\delRIP}{1-\delRIP}} \|\ve{x}\|_2 \geq \sqrt{\frac{1-2\delRIP}{1-\delRIP}} \frac{1}{\sqrt{1+\delRIP}}\|\ve{v}\|_2$.
\end{proof}
\begin{corollary}\label{corr:bp2}
If $\ve{x} \in \mathbb{R}^{|S_j|}$, then for the i.i.d. Gaussian signal model, 
$\ve{x}^T \ve{A}_j^T \Ph(S_i)^{-1} \ve{A}_j \ve{x} \geq \frac{1-2\delRIP}{1-\delRIP} \frac{\|\ve{x}\|_2^2}{\sigma_e^2}$. 
Also the singular values of $\ve{A}_j^T \Ph(S_i)^{-1} \ve{A}_j$ are greater than or equal to $\frac{1-2\delRIP}{1-\delRIP} \frac{1}{\sigma_e^2}$.
\end{corollary}
\begin{proof}
We note that,
\begin{eqnarray}
&&\ve{\Phi}(S_j)=\sigma_1^2\ve{A}_j\ve{A}_j^T + \sigma_e^2\ve{I}_M  
=\ve{U}_j (\sigma_1^2 \ve{\Sigma}_j\ve{\Sigma}_j^T + \sigma_e^2\ve{I}_M)\ve{U}_j^T, \\
&&\mathrm{hence,} \qquad \ve{\Phi}(S_j)^{-1}=\ve{U}_j (\sigma_1^2 \ve{\Sigma}_j\ve{\Sigma}_j^T + \sigma_e^2\ve{I}_M)^{-1}\ve{U}_j^T, \label{eqn:Phis0}
\end{eqnarray}
and $\ve{A}_j^T \Ph(S_i)^{-1} \ve{A}_j$ is a symmetric and positive definite matrix. Thus
\begin{eqnarray}
&&\ve{x}^T \ve{A}_j^T \Ph(S_i)^{-1} \ve{A}_j \ve{x} =  \ve{x}^T \ve{A}_j^T \ve{U}_i (\sigma_1^2 \ve{\Sigma}_i \ve{\Sigma}_i^T + \sigma_e^2 \ve{I}_M)^{-1} \ve{U}_i^T \ve{A}_j \ve{x} \\
&& \qquad = \ve{x}^T \ve{A}_j^T \bar{\ve{U}}_i (\sigma_1^2 \bar{\ve{\Sigma}}_i \bar{\ve{\Sigma}}_i^T + \sigma_e^2 \ve{I}_{|S_i|})^{-1} \bar{\ve{U}}_i^T \ve{A}_j \ve{x} + \ve{x}^T \ve{A}_j^T \underbar{\ve{U}}_i ( \sigma_e^2 \ve{I}_{M-|S_i|})^{-1} \underbar{\ve{U}}_i^T \ve{A}_j \ve{x} \\
&& \qquad \geq  \ve{x}^T \ve{A}_j^T \underbar{\ve{U}}_i ( \sigma_e^2 \ve{I}_{M-|S_i|})^{-1} \underbar{\ve{U}}_i^T \ve{A}_j \ve{x} = \frac{1}{\sigma_e^2} \| \underbar{\ve{U}}_i^T \ve{A}_j \ve{x} \|_2^2 \geq \frac{1-2\delRIP}{1-\delRIP} \frac{\|\ve{x}\|_2^2}{\sigma_e^2}. \label{eqn:corr_ineq}
\end{eqnarray}
The last inequality follows from proposition~\ref{prop:bp1}. Since $\ve{A}_j^T \Ph(S_i)^{-1} \ve{A}_j$ is symmetric and positive definite, it has SVD 
$\ve{A}_j^T \Ph(S_i)^{-1} \ve{A}_j = \ve{U} \ve{\Sigma} \ve{U}^T$. Let the $k^\mathrm{th}$ singular value be $\sigma_k$ and 
the singular vector corresponding to the singular value $\sigma_k$ be $\ve{u}_k \in \mathbb{R}^{|S_j|}$. Then
\begin{eqnarray}
\ve{u}_k^T \ve{A}_j^T \Ph(S_i)^{-1} \ve{A}_j \ve{u}_k= \ve{u}_k^T \ve{U} \ve{\Sigma} \ve{U}^T \ve{u}_k = \sigma_k. \label{eqn:corr_smin}
\end{eqnarray}
Since $\|\ve{u}_k\|_2^2 =1$, from~(\ref{eqn:corr_ineq}) and~(\ref{eqn:corr_smin}) it follows that $\sigma_k \geq \frac{1-2\delRIP}{1-\delRIP} \frac{1}{\sigma_e^2}$ and this
is true for any $k$. 
\end{proof}
The next proposition is about the tail probability bound of the Chi-squared distribution.
\begin{proposition}\label{BMAP_proposition:chi_sq}
Suppose $n$ independent and identically distributed variables $X_i \sim \mathcal{N}(0,\sigma^2)$. 
If Chi-squared distributed random variable $Z=\sum_{i=1}^n X_i^2$, then for any $\beta > 1$,
\begin{eqnarray}
&&\mathbb{P} \left[Z>\beta n\sigma^2\right] \leq e^{-\frac{n}{2} (\beta-1 - \ln \beta)}. \label{eqn:BMAP:chi_sq} 
\end{eqnarray}
\end{proposition}
\begin{proof}
Let $\bar{X}_i=\frac{X_i}{\sigma}$. Then $\bar{X}_i \sim \mathcal{N}(0,1)$ and are independently distributed.
Let $\bar{Z} = \sum_{i=1}^n \bar{X}_i^2 =\frac{Z}{\sigma^2}$ which is Chi-squared distributed with degree of freedom $n$.
Using Chernoff inequality,
\begin{eqnarray}
\mathbb{P} \left[Z>\beta n\sigma^2\right] &=& \mathbb{P} \left[ e^{t\bar{Z}} > e^{n \beta t} \right], \qquad \mathrm{for~any}~~ t>0 \\
&\leq& \frac{\mathbb{E} \left[ e^{t \bar{Z}}\right ]}{e^{n \beta t}} = \frac{ \prod_{i=1}^n \mathbb{E} \left[ e^{t \bar{X}_i^2}\right ]}{e^{n \beta t}} = \frac{(1-2t)^{-\frac{n}{2}}}{e^{n \beta t}},  \qquad \mathrm{for}~~  t \in (0, 1/2) \nonumber\\
&=& e^{-\frac{n}{2}(\ln(1-2t)+2\beta t)}.
\end{eqnarray}
The minimum is attained at $t=\frac{\beta-1}{2\beta}$ which gives inequality~(\ref{eqn:BMAP:chi_sq}).
\end{proof}
We also use the following inequality at various places. If $c,d >0$, then
\begin{eqnarray}
\frac{(a+b)^2}{c+d} \leq \frac{(a+b)^2 + (a-b)^2}{c+d} =  \frac{2a^2}{c+d} +  \frac{2b^2}{c+d} < \frac{2a^2}{c} + \frac{2b^2}{d}.
\end{eqnarray}

\subsection{Proof of Theorem~\ref{th:BMAP1}} \label{proof:BMAP1}
Let us divide the indices for the columns of the \ve{A} matrix into four disjoint subsets $S_0$, $S_1$, $S_2$ and 
$S_3$ such that $S_0$ denotes the columns which are in the true support and are correctly identified by the 
constrained MAP estimator $\widehat{S}$, $S_1$ denotes the missed columns, $S_2$ denotes the columns which are not 
in the true support but selected by $\widehat{S}$,  and $S_3$ denotes the columns which are neither in true support nor 
in $\widehat{S}$. Define $S_{ij}=S_i \cup S_j$. Let $\ve{A}_{ij}$ denote the matrix consisting of those columns 
of \ve{A} which are indexed by the set $S_{ij}$. Thus,
\begin{eqnarray}
\ve{y} = \ve{A}_{01} \ve{x}_{01} + \ve{e} = \mu_1 \ve{A}_{01} \ve{1}_{01} + \ve{A}_{01} \ve{z}_{01} + \ve{e}, \label{eqn:y}
\end{eqnarray}
where $\ve{z}_{01} \sim \mathcal{N}(\ve{0},\sigma_1^2 \ve{1}_{|S_{01}|})$. For zero mean model, 
$\mu_1=0$ and $\ve{z}_{01}=\ve{x}_{01}$.

We have defined the event~$E$ to be $|S_{01}| \leq 2Np$. The mean value of $|S_{01}|$ is $\mathbb{E}[|S_{01}|]=Np$. 
Using Chernoff bound on upper tail of Binomial distribution~\cite[pp.~68]{ref:RandAlg}, 
\begin{eqnarray}
\mathbb{P}[|S_{01}| > (1+ \delta) \mathbb{E}[|S_{01}|] ] < \left(\frac{e^\delta}{(1+\delta)^{(1+\delta)}} \right)^{\mathbb{E}[|S_{01}|]}.
\end{eqnarray}
Taking $\delta=1$,
\begin{eqnarray}
\mathbb{P}[E] = \mathbb{P}[|S_{01}| \leq 2Np] > 1-e^{-Np(2 \ln2 -1)}.
\end{eqnarray}
If $E^c$ denotes the complement of $E$ {\em i.e.,} the event $|S| > 2Np$, then for any event $B$,
\begin{eqnarray}
\mathbb{P}[B] = \mathbb{P} [E] \mathbb{P} [B|E] + \mathbb{P} [E^c] \mathbb{P} [B|E^c]  \geq \mathbb{P} [E] \mathbb{P} [B|E].
\end{eqnarray}
For the rest of the proof we assume that event $E$ holds and all the 
subsequent probabilities are conditioned on event~$E$.

For convenience we define the function to be minimized in~(\ref{eqn:MAP2}) as $\gamma(S)$ {\em i.e.,}  
\begin{eqnarray}
\gamma(S) &=& \frac{1}{2} \ln \det (\Phs) + \frac{1}{2} (\ve{y}-\mu_1 \ve{A}_S\ve{1}_{|S|})^T\Phsi (\ve{y}-\mu_1 \ve{A}_S\ve{1}_{|S|}) +  |S| \ln \frac{1-p}{p}  \label{eqn:gammas} \\
&=& \frac{1}{2}\gamma_1(S) + \frac{1}{2}\gamma_2(S) + \gamma_3(S),  \label{eqn:gamma3s}
\end{eqnarray}
where $\gamma_1(S)=\ln \det (\Phs)$, $\gamma_2(S)=(\ve{y}-\mu_1 \ve{A}_S\ve{1}_{|S|})^T\Phsi (\ve{y}-\mu_1 \ve{A}_S\ve{1}_{|S|})$ and $\gamma_3(S)=|S| \ln \frac{1-p}{p}$.

Let the SVD of $\ve{A}_{0}$ be $\ve{A}_{0}=\ve{U}_{0} \ve{\Sigma}_{0} \ve{V}_{0}^T$. Let $\bar{\ve{U}}_{0}$ denote the submatrix of $\ve{U}_{0}$ consisting of the 
first $|S_{0}|$ columns and $\underbar{\ve{U}}_{0}$ denote the submatrix with the rest of the columns.
Thus $\bar{\ve{U}}_{0}$ forms an orthonormal basis for the column space $\mathcal{A}_{0}$ of $\ve{A}_{0}$.
$\underbar{\ve{U}}_{0}$ forms an orthonormal basis for the space 
$\mathbb{R}^M \setminus \mathcal{A}_{0}$. Let $\bar{\ve{\Sigma}}_{0}$ denote the $|S_{0}| \times |S_{0}|$
upper left square submatrix of $\ve{\Sigma}_{0}$. From~(\ref{eqn:Phs}),
\begin{eqnarray}
\Ph(S_{01})&=&\Ph(S_{0}) + \sigma_1^2 \ve{A}_1 \ve{A}_1^T.
\end{eqnarray}
Hence applying matrix determinant lemma,
\begin{eqnarray}
\gamma_1(S_{01}) &=&\ln \det(\Ph(S_{01}))=\ln \det(\Ph(S_{0})) + \ln \det ( \ve{I}_{|S_1|} + \sigma_1^2 \ve{A}_1^T \Ph(S_{0})^{-1} \ve{A}_1) \\
&=& \ln \det(\Ph(S_{0})) + \ln \det ( \ve{I}_{|S_1|} + \sigma_1^2 \ve{A}_1^T   \ve{U}_0 (\sigma_1^2 \ve{\Sigma}_0 \ve{\Sigma}_0^T + \sigma_e^2 \ve{I}_M)^{-1} \ve{U}_0^T \ve{A}_1) \\
&\leq& \ln \det(\Ph(S_{0})) + |S_1|\ln \left( 1 + \frac{\sigma_1^2}{\sigma_e^2} (1+\delRIP)\right). \label{eqn:Bth1:lndetS01} 
\end{eqnarray}
The inequality in~(\ref{eqn:Bth1:lndetS01}) follows from the facts that the maximum singular value of the matrix $\ve{A}_1$ is $\sqrt{1+\delRIP}$ and maximum
value on the diagonal of the diagonal matrix $(\sigma_1^2 \ve{\Sigma}_0 \ve{\Sigma}_0^T + \sigma_e^2 \ve{I}_M)^{-1}$ is $\frac{1}{\sigma_e^2}$ and
$\sigma_1^2 \ve{A}_1^T   \ve{U}_0 (\sigma_1^2 \ve{\Sigma}_0 \ve{\Sigma}_0^T + \sigma_e^2 \ve{I}_M)^{-1} \ve{U}_0^T \ve{A}_1$, being a  symmetric and positive definite matrix, has
SVD of the form $\ve{U} \ve{\Sigma} \ve{U}^T$.
A lower bound on $\gamma_1(S_{02})$ can be obtained proceeding 
in a similar way as~(\ref{eqn:Bth1:lndetS01}) was obtained but taking lower bound instead of upper bound. We note that from corollary~\ref{corr:bp2}, the minimum singular value 
of $\sigma_1^2 \ve{A}_2^T \Ph(S_{0})^{-1} \ve{A}_2$ is at least $\frac{1-2\delRIP}{1-\delRIP} \frac{\sigma_1^2}{\sigma_e^2}$. Thus
\begin{eqnarray}
\gamma_1(S_{02})=\ln \det (\Ph(S_{02})) \geq  \ln \det(\Ph(S_{0})) + |S_2|\left( 1 + \frac{\sigma_1^2}{\sigma_e^2} \left(\frac{1- 2\delRIP}{1-\delRIP}\right)\right). \label{eqn:Bth1:lndetS02} 
\end{eqnarray}
Let the SVD of $\ve{A}_{01}$ be $\ve{A}_{01}=\ve{U}_{01}\ve{\Sigma}_{01} \ve{V}_{01}^T$. 
Let $\bar{\ve{U}}_{01}$ denote the submatrix of $\ve{U}_{01}$ consisting of the 
first $|S_{01}|$ columns and $\underbar{\ve{U}}_{01}$ denote the submatrix with the rest of the columns.
Thus $\bar{\ve{U}}_{01}$ forms an orthonormal basis for the column space $\mathcal{A}_{01}$ of $\ve{A}_{01}$.
$\underbar{\ve{U}}_{01}$ forms an orthonormal basis for the space 
$\mathbb{R}^M \setminus \mathcal{A}_{01}$.  
The measured data $\ve{y}$ is noisy linear combination of the columns of $\ve{A}$ selected by $S_{01}$.
\begin{eqnarray}
\ve{y} - \mu_1 \ve{A}_{01} \ve{1}_{01} =\ve{A}_{01}\ve{z}_{01} +\ve{e}
=\ve{U}_{01}\ve{\Sigma}_{01}\ve{V}_{01}^T\ve{z}_{01} + \ve{e} \label{eqn:ySVD}    
\end{eqnarray}
Let $\ve{e} = \bar{\ve{U}}_{01}\bar{\ve{e}}_{01} +\underbar{\ve{U}}_{01}\underbar{\ve{e}}_{01}$. Thus 
from~(\ref{eqn:Phis0}) and (\ref{eqn:ySVD}) and the fact that 
$\underbar{\ve{U}}_{01}^T\ve{A}_{01}\ve{z}_{01} = \ve{0}_{M-|S_{01}|}$,
\begin{eqnarray}
\gamma_2(S_{01})&=&(\ve{y}- \mu_1 \ve{A}_{01} \ve{1}_{01})^T \Ph(S_{01})^{-1} (\ve{y}- \mu_1 \ve{A}_{01} \ve{1}_{01}) \\
&=&(\ve{U}_{01} \ve{\Sigma}_{01} \ve{V}_{01}^T\ve{z}_{01} + \ve{e})^T\ve{U}_{01} (\sigma_1^2 \ve{\Sigma}_{01} \ve{\Sigma}_{01}^T + \sigma_e^2 \ve{I}_M)^{-1} \ve{U}_{01}^T (\ve{U}_{01} \ve{\Sigma}_{01} \ve{V}_{01}^T\ve{z}_{01}+\ve{e}) \\
&=&(\bar{\ve{\Sigma}}_{01} \ve{V}_{01}^T\ve{z}_{01} + \bar{\ve{e}}_{01})^T  (\sigma_1^2 \bar{\ve{\Sigma}}_{01} \bar{\ve{\Sigma}}_{01}^T + \sigma_e^2 \ve{I}_{01})^{-1}  (\bar{\ve{\Sigma}}_{01} \ve{V}_{01}^T\ve{z}_{01}+\bar{\ve{e}}_{01}) + \frac{1}{\sigma_e^2} \| \underbar{\ve{e}}_{01}\|_2^2 \\ 
&\leq& \frac{(\sqrt{1+\delRIP} \|\ve{z}_{01}\|_2 + \|\bar{\ve{e}}_{01} \|_2)^2}{(1-\delRIP) \sigma_1^2 + \sigma_e^2}   + \frac{\|\underbar{\ve{e}}_{01}\|_2^2}{\sigma_e^2}  \\
&<& \frac{2 (1+\delRIP)\|\ve{z}_{01}\|_2^2 }{(1-\delRIP)\sigma_1^2 }   + \frac{2 \|\bar{\ve{e}}_{01} \|_2^2}{\sigma_e^2}   + \frac{\|\underbar{\ve{e}}_{01}\|_2^2}{\sigma_e^2}. \label{eqn:Bth1:g2S01}
\end{eqnarray}
Now we obtain a lower bound on $\gamma_2(S_{02})$. Let the SVD of 
$\ve{A}_{02}$ be $\ve{A}_{02}=\ve{U}_{02} \ve{\Sigma}_{02} \ve{V}_{02}^T$. 
Let $\bar{\ve{U}}_{02}$, $\underbar{\ve{U}}_{02}$ and $\bar{\ve{\Sigma}}_{02}$ be defined similar to 
$\bar{\ve{U}}_{01}$, $\underbar{\ve{U}}_{01}$ and $\bar{\ve{\Sigma}}_{01}$ respectively.
Let $\bar{\ve{W}}_{1 \setminus 02}$ be an orthonormal basis 
for the subspace spanned by $\underbar{\ve{U}}_{02} \underbar{\ve{U}}_{02}^T \bar{\ve{U}}_1$. Let us denote this 
subspace by $\sp{A}_{1 \setminus 02}$. Also, let $\bar{\ve{U}}_{012}$ be an orthonormal basis for the column space
of~$\ve{A}_{012}$ and $\underbar{\ve{U}}_{012}$ be an orthonormal basis for the left null space $\mathbb{R}^M \setminus \sp{A}_{012}$.
The two subspaces $\sp{A}_{1 \setminus 02}$ and $\mathbb{R}^M \setminus \sp{A}_{012}$ are orthogonal and their union 
is the subspace $\mathbb{R}^M \setminus \sp{A}_{02}$. Now,
\begin{eqnarray}
\gamma_2(S_{02})&=&(\ve{y} - \mu_1 \ve{A}_{02} \ve{1}_{02})^T\Ph(S_{02})^{-1} (\ve{y}-\mu_1 \ve{A}_{02}\ve{1}_{02})\\ 
&=&(\ve{y} - \mu_1 \ve{A}_{02} \ve{1}_{02})^T \ve{U}_{02} (\sigma_1^2 \ve{\Sigma}_{02} \ve{\Sigma}_{02}^T + \sigma_e^2 \ve{I}_M)^{-1} \ve{U}_{02}^T(\ve{y}-\mu_1 \ve{A}_{02}\ve{1}_{02})\\ 
&=&(\ve{y} - \mu_1 \ve{A}_{02} \ve{1}_{02})^T \bar{\ve{U}}_{02} (\sigma_1^2 \bar{\ve{\Sigma}_{02}} \bar{\ve{\Sigma}}_{02}^T + \sigma_e^2 \ve{I}_{02})^{-1} \bar{\ve{U}}_{02}^T(\ve{y}-\mu_1 \ve{A}_{02}\ve{1}_{02}) \label{eqn:Bth1:2nd02ineq_pre} \nonumber \\
&& \quad \quad \quad \quad \quad \quad \quad \quad + \frac{1}{\sigma_e^2} \|\underbar{\ve{U}}_{02}^T(\ve{y}-\mu_1 \ve{A}_{02}\ve{1}_{02})\|_2^2 \label{eqn:Bth1:2nd02ineq}\\
&\geq& \frac{1}{\sigma_e^2} \|\underbar{\ve{U}}_{02}^T(\ve{y}-\mu_1 \ve{A}_{02}\ve{1}_{02}) \|_2^2  \label{eqn:Bth1:2nd03ineq}\\
&=& \frac{1}{\sigma_e^2} \|\underbar{\ve{U}}_{02}^T (\ve{A}_0 \ve{x}_0 + \ve{A}_1 \ve{x}_1 + \ve{e}-\mu_1 \ve{A}_{02}\ve{1}_{02})\|_2^2 = \frac{1}{\sigma_e^2} \|\underbar{\ve{U}}_{02}^T (\ve{A}_1 \ve{x}_1 + \ve{e})\|_2^2 \\
&=& \frac{1}{\sigma_e^2} \|\bar{\ve{W}}_{1 \setminus 02}^T (\ve{A}_1 \ve{x}_1 + \ve{e})\|_2^2 + \frac{1}{\sigma_e^2} \|\underbar{\ve{U}}_{102}^T (\ve{A}_1 \ve{x}_1 + \ve{e})\|_2^2 \\
&=& \frac{1}{\sigma_e^2} \|\underbar{\ve{U}}_{02}^T \ve{A}_1 \ve{x}_1 + \bar{\ve{W}}_{1 \setminus 02}^T \ve{e}\|_2^2 + \frac{1}{\sigma_e^2} \|\underbar{\ve{U}}_{102}^T \ve{e}\|_2^2 
\end{eqnarray}
Now from proposition~\ref{prop:bp1}, $\|\underbar{\ve{U}}_{02}^T \ve{A}_1 \ve{x}_1\|_2 \geq \sqrt{\frac{1-2\delRIP}{1-\delRIP}} \| \ve{x}_1 \|_2$. 
We assume that $ \sqrt{\frac{1-2\delRIP}{1-\delRIP}} \| \ve{x}_1 \|_2 > \|\bar{\ve{W}}_{1 \setminus 02}^T \ve{e}\|_2$.
Otherwise there is nothing left to prove. We note that 
$\|\bar{\ve{W}}_{1 \setminus 02}^T \ve{e}\|_2= \|\bar{\ve{e}}_{1 \setminus 02} \|_2 \leq \|\bar{\ve{e}}_1\|_2$.
Thus,
\begin{eqnarray}
\gamma_2(S_{02}) \geq \frac{1}{\sigma_e^2} \left( \sqrt{\frac{1-2\delRIP}{1-\delRIP}} \| \ve{x}_1 \|_2 - \|\bar{\ve{e}}_1\|_2\right)^2  + \frac{1}{\sigma_e^2} \|\underbar{\ve{e}}_{012}\|_2^2. \label{eqn:Bth1:prelb2nd1} 
\end{eqnarray}
Also, 
\begin{eqnarray}
\gamma_3(S_{01}) - \gamma_3(S_{02}) = (|S_1| - |S_2|) \ln \frac{1-p}{p}.   \label{eqn:Bth1:g3S0102}
\end{eqnarray}
Now since $S_{02}=\widehat{S}$, $\gamma(S_{02}) \leq \gamma(S_{01})$. Thus, from~(\ref{eqn:Bth1:lndetS01}), (\ref{eqn:Bth1:lndetS02}), (\ref{eqn:Bth1:g2S01}), (\ref{eqn:Bth1:prelb2nd1}) and~(\ref{eqn:Bth1:g3S0102}),
\begin{eqnarray}
&&|S_2| \left(  \ln \left(1+ \frac{\sigma_1^2}{\sigma_e^2} \left( \frac{1-2\delRIP}{1-\delRIP}\right) \right) + 2 \ln \frac{1-p}{p}\right)
+ \frac{\left( \sqrt{\frac{1-2\delRIP}{1-\delRIP}} \| \ve{x}_1 \|_2 - \|\bar{\ve{e}}_1\|_2\right)^2}{\sigma_e^2} + \frac{ \|\underbar{\ve{e}}_{012}\|_2^2}{\sigma_e^2}\nonumber\\
&& \qquad \leq |S_1| \left(  \ln \left(1+ \frac{\sigma_1^2}{\sigma_e^2} (1 + \delRIP) \right) + 2 \ln \frac{1-p}{p}\right)  + \frac{\|\underbar{\ve{e}}_{01}\|_2^2}{\sigma_e^2} + \frac{2 (1+\delRIP)\|\ve{z}_{01}\|_2^2 }{(1-\delRIP)\sigma_1^2 }   + \frac{2 \|\bar{\ve{e}}_{01} \|_2^2}{\sigma_e^2}. \label{eqn:Bth1:ineq1}
\end{eqnarray}
Since $|S_2|  \geq 0$ and $p < 1/2$ for sparse signals, the first term is non-negative. Hence,
\begin{eqnarray}
\frac{\left( \sqrt{\frac{1-2\delRIP}{1-\delRIP}} \| \ve{x}_1 \|_2 - \|\bar{\ve{e}}_1\|_2\right)^2}{\sigma_e^2}    &\leq& |S_1| \left(  \ln \left(1+ \frac{\sigma_1^2}{\sigma_e^2} (1 + \delRIP) \right) + 2 \ln \frac{1-p}{p}\right)  + \frac{2 (1+\delRIP)\|\ve{z}_{01}\|_2^2 }{(1-\delRIP)\sigma_1^2 }   \nonumber\\
&&\qquad+ \frac{2 \|\bar{\ve{e}}_{01} \|_2^2}{\sigma_e^2} + \frac{\|\underbar{\ve{e}}_{01}\|_2^2 - \|\underbar{\ve{e}}_{012}\|_2^2}{\sigma_e^2}. \label{eqn:Bth1:ineq2} 
\end{eqnarray}
Now $\|\underbar{\ve{e}}_{01}\|_2^2 - \|\underbar{\ve{e}}_{012}\|_2^2 \leq \|\bar{\ve{e}}_{2 \setminus 01}\|_2^2 \leq \|\bar{\ve{e}}_2\|_2^2$. 
Now consider the expression $\frac{ \|\bar{\ve{e}}_{01} \|_2^2}{\sigma_e^2}$. Note that 
$\bar{\ve{e}}_{01} = \bar{\ve{U}}_{01}^T \ve{e} $ is the projection of $\ve{e}$ onto the $|S_{01}|$-dimensional 
subspace~$\sp{A}_{01}$. Thus $\bar{\ve{e}}_{01} \sim \mathcal{N}(\ve{0} , \sigma_e^2 \ve{I}_{01})$. Let
$\tilde{\ve{e}}_{01}^T = [\bar{\ve{e}}_{01}^T, \bar{e}_1, \bar{e}_2, \ldots, \bar{e}_{\lfloor 2Np-|S_{01}| \rfloor}]$ such that
$\tilde{\ve{e}}_{01} \sim \mathcal{N}(\ve{0} , \sigma_e^2 \ve{I}_{\lfloor 2Np \rfloor})$. By proposition~\ref{BMAP_proposition:chi_sq}, 
$\| \bar{\ve{e}}_{01} \|_2^2 \leq \| \tilde{\ve{e}}_{01} \|_2^2 \leq 2 \beta Np \sigma_e^2$ with probability exceeding
$1-e^{-Np(\beta-1 - \ln \beta)}$ for $\beta>1$. Similarly $\| \bar{\ve{z}}_{01} \|_2^2 \leq 2 \beta Np \sigma_1^2$ with
probability exceeding $1-e^{-Np(\beta-1 - \ln \beta)}$ and  $\| \bar{\ve{e}}_{2} \|_2^2 \leq 2 \beta Np \sigma_e^2$
probability exceeding $1-e^{-Np(\beta-1 - \ln \beta)}$. Therefore with probability exceeding $1-3 e^{- Np(\beta-1 - \ln \beta)}$,
\begin{eqnarray}
\frac{\left( \sqrt{\frac{1-2\delRIP}{1-\delRIP}} \| \ve{x}_1 \|_2 - \|\bar{\ve{e}}_1\|_2\right)^2}{\sigma_e^2} &\leq& C |S_1| + 4 \left( \frac{1+\delRIP}{1-\delRIP}\right)\beta Np + 4 \beta Np + 2 \beta Np \label{eqn:Bth1:ineq4}\\
&\leq&  2C Np + 8 \beta Np + 4 \beta Np + 2 \beta Np = (14 \beta + 2C)Np 
\end{eqnarray}
since $\delRIP \leq 1/3$. 
 Thus,
\begin{eqnarray}
\sqrt{\frac{1-2\delRIP}{1-\delRIP}} \| \ve{x}_1 \|_2 \leq  \|\bar{\ve{e}}_1\|_2 + \sqrt{(14 \beta + 2C) Np} \sigma_e \leq(\sqrt{2 \beta} + \sqrt{14 \beta + 2C}) \sqrt{Np} \sigma_e. \label{eqn:Bth1:ineq3} 
\end{eqnarray}
Since $\delRIP \leq 1/3$, we can write~(\ref{eqn:Bth1:ineq3}) as
$\| \ve{x}_1 \|_2 \leq  \sqrt{2} (\sqrt{2 \beta} + \sqrt{14 \beta + 2C}) \sqrt{Np} \sigma_e$. This holds with overall
probability exceeding $(1-e^{-Np(2 \ln2 -1)}) (1-3 e^{- Np(\beta-1 - \ln \beta)})$ for $\beta>1$.  \qed

\subsection{Proof of Theorem~\ref{th:BMAP2}} \label{proof:BMAP2}
Similar to theorem~\ref{th:BMAP1} we assume that event $E$ {\em i.e.,} $|S_{01}| \leq 2Np$ is true. This holds with probability
exceeding $1-e^{-Np(2 \ln2 -1)}$. For the rest of the proof all events and probabilities are conditioned on this event. 

Here we show that if $\mu_1$ and $\sigma_1$ satisfy the condition stated in theorem~\ref{th:BMAP2}, then $\gamma(S_{02})$ cannot be smaller or equal 
to $\gamma(S_{01})$ unless $S_1=S_2=\emptyset$. 
We obtained upper bound on $\gamma(S_{01})$ and lower bound on $\gamma(S_{02})$ in the proof of theorem~\ref{th:BMAP1}. 
If the lower bound is greater than the upper bound then we reach a contradiction that 
$\gamma(S_{02})$ cannot be the estimate of~$S$. This happens when the inequality in~(\ref{eqn:Bth1:ineq1}) is 
reversed, {\em i.e.,} if 
\begin{eqnarray}
&&|S_2| \left(  \ln \left(1+ \frac{\sigma_1^2}{\sigma_e^2} \left( \frac{1-2\delRIP}{1-\delRIP}\right) \right) + 2 \ln \frac{1-p}{p}\right)
+ \frac{\left( \sqrt{\frac{1-2\delRIP}{1-\delRIP}} \| \ve{x}_1 \|_2 - \|\bar{\ve{e}}_1\|_2\right)^2}{\sigma_e^2} + \frac{ \|\underbar{\ve{e}}_{012}\|_2^2}{\sigma_e^2}\nonumber\\
&& \qquad > |S_1| \left(  \ln \left(1+ \frac{\sigma_1^2}{\sigma_e^2} (1 + \delRIP) \right) + 2 \ln \frac{1-p}{p}\right)  + \frac{\|\underbar{\ve{e}}_{01}\|_2^2}{\sigma_e^2} + \frac{2 (1+\delRIP)\|\ve{z}_{01}\|_2^2 }{(1-\delRIP)\sigma_1^2 }   + \frac{2 \|\bar{\ve{e}}_{01} \|_2^2}{\sigma_e^2}. \label{eqn:Bth2:ineq1}
\end{eqnarray}
Thus the following inequality is sufficient for~(\ref{eqn:Bth2:ineq1}) to be true.
\begin{eqnarray}
\frac{\left( \sqrt{\frac{1-2\delRIP}{1-\delRIP}} \| \ve{x}_1 \|_2 - \|\bar{\ve{e}}_1\|_2\right)^2}{\sigma_e^2} + \frac{ \|\underbar{\ve{e}}_{012}\|_2^2}{\sigma_e^2}
&>& |S_1| \left(  \ln \left(1+ \frac{\sigma_1^2}{\sigma_e^2} (1 + \delRIP) \right) + 2 \ln \frac{1-p}{p}\right)  \nonumber\\
&& + \frac{\|\underbar{\ve{e}}_{01}\|_2^2}{\sigma_e^2} + \frac{2 (1+\delRIP)\|\ve{z}_{01}\|_2^2 }{(1-\delRIP)\sigma_1^2 }   + \frac{2 \|\bar{\ve{e}}_{01} \|_2^2}{\sigma_e^2}. \label{eqn:Bth2:ineq2}
\end{eqnarray}
This is equivalent to,
\begin{eqnarray}
\frac{\left( \sqrt{\frac{1-2\delRIP}{1-\delRIP}} \| \ve{x}_1 \|_2 - \|\bar{\ve{e}}_1\|_2\right)^2}{\sigma_e^2}
&>& |S_1| \left(  \ln \left(1+ \frac{\sigma_1^2}{\sigma_e^2} (1 + \delRIP) \right) + 2 \ln \frac{1-p}{p}\right)   \nonumber\\
&& + \frac{\|\underbar{\ve{e}}_{01}\|_2^2 -  \|\underbar{\ve{e}}_{012}\|_2^2}{\sigma_e^2} + \frac{2 (1+\delRIP)\|\ve{z}_{01}\|_2^2 }{(1-\delRIP)\sigma_1^2 }   + \frac{2 \|\bar{\ve{e}}_{01} \|_2^2}{\sigma_e^2}. \label{eqn:Bth2:ineq3}
\end{eqnarray}
We have seen in the proof of theorem~\ref{th:BMAP1} that the right hand side is bounded above by $(14 \beta + 2C)Np$ with probability exceeding 
$1-3 e^{- Np(\beta-1 - \ln \beta)}$. Thus if
\begin{eqnarray}
\| \ve{x}_1 \|_2 >  \sqrt{2} (\sqrt{2 \beta} + \sqrt{14 \beta + 2C}) \sqrt{Np} \sigma_e, \label{eqn:Bth2:ineq4}
\end{eqnarray}
 (\ref{eqn:Bth2:ineq3})~is satisfied with probability exceeding $1-3 e^{- Np(\beta-1 - \ln \beta)}$.
Now for sufficiently large $|\mu_1|$, $\| \ve{x}_1 \|_2 = \|\mu_1 \ve{1}_1 + \ve{z}_1 \|_2 \geq  \|\mu_1 \ve{1}_1 \|_2 - \| \ve{z}_1 \|_2 \geq (|\mu_1| - \sqrt{\bar{\beta}}\sigma_1) \sqrt{|S_1|}$
with probability exceeding $1-e^{- \frac{|S_1|(\bar{\beta}-1 - \ln \bar{\beta})}{2}}$. If $|S_1| \geq 1$, then $\gamma(S_{02})$ becomes 
greater than $\gamma(S_{01})$ with probability exceeding
$1-3 e^{- Np(\beta-1 - \ln \beta)} - e^{- \frac{(\bar{\beta}-1 -\ln \bar{\beta})}{2}}$ if,
\begin{eqnarray}
|\mu_1| >  \sqrt{\bar{\beta}} \sigma_1 + \sqrt{2} (\sqrt{2 \beta} + \sqrt{14 \beta + 2C}) \sqrt{Np} \sigma_e.\label{eqn:Bth2:ineq5}
\end{eqnarray}
Hence $|S_1|=0$ {\em i.e.,} the set~$S_1$ is empty and $S_{01}=S_0$. 
Thus~(\ref{eqn:Bth2:ineq5}) is a probabilistic sufficient condition that no active coefficient is missing.
Now we assume that~(\ref{eqn:Bth2:ineq5}) is satisfied and we investigate what (additional) condition
guarantees no false alarm with very high probability. We assume $S_{2}$ is not empty and find out the condition on $\mu_1$ and $\sigma_1$ that contradicts this 
assumption.
\begin{eqnarray}
\gamma_1(S_{02})  &\geq& \gamma_1(S_{01}) + |S_2| \left(1+\frac{\sigma_1^2}{\sigma_e^2} \left(\frac{1-2\delRIP}{1-\delRIP} \right) \right),\label{eqn:BMAP:th2:g1} \\
\mathrm{and,} \quad \gamma_3(S_{02})&=& \gamma_3(S_{01}) + |S_2| \ln \frac{1-p}{p}. \label{eqn:BMAP:th2:g3} 
\end{eqnarray}
Since set $S_1$ is empty,
\begin{eqnarray}
\gamma_1(S_{01}) \leq \frac{\|\bar{\ve{e}}_0 + \ve{A}_0 \ve{z}_0 \|_2^2}{(1-\delRIP) \sigma_1^2 + \sigma_e^2} + \frac{\|\underbar{\ve{e}}_0\|_2^2}{\sigma_e^2}.
\end{eqnarray}
In obtaining~(\ref{eqn:Bth1:2nd03ineq}) from~(\ref{eqn:Bth1:2nd02ineq_pre}) we lower bounded the first term by zero. Now we use a tighter 
lower bound by explicitly using the condition that $\mu_1 \neq 0$. 
\begin{eqnarray}
\gamma_2(S_{02})&=&(\ve{y} - \mu_1 \ve{A}_{02} \ve{1}_{02})^T \bar{\ve{U}}_{02} (\sigma_1^2 \bar{\ve{\Sigma}_{02}} \bar{\ve{\Sigma}}_{02}^T + \sigma_e^2 \ve{I}_{02})^{-1} \bar{\ve{U}}_{02}^T(\ve{y}-\mu_1 \ve{A}_{02}\ve{1}_{02}) \nonumber \\
&& \quad \quad \quad \quad \quad \quad \quad \quad + \frac{1}{\sigma_e^2} \|\underbar{\ve{U}}_{02}^T(\ve{y}-\mu_1 \ve{A}_{02}\ve{1}_{02})\|_2^2. \label{eqn:Bth2:2nd02ineq}
\end{eqnarray}
Here $\ve{y}=\ve{A}_0 \ve{x}_0 + \ve{e}$. Thus $\underbar{\ve{U}}_{02}^T(\ve{y}-\mu_1 \ve{A}_{02}\ve{1}_{02})=\underbar{\ve{U}}_{02}^T \ve{e} =\underbar{\ve{e}}_{02}$.
Let $\bar{\ve{W}}_{0 \setminus 2}$ and $\bar{\ve{W}}_{2 \setminus 0}$ be the orthonormal bases for the orthogonal subspaces $\sp{A}_0 \setminus \sp{A}_2$ and $\sp{A}_2 \setminus \sp{A}_0$
respectively. Thus,
\begin{eqnarray}
\gamma_2(S_{02}) &\geq& \frac{ \| \bar{\ve{U}}_{02}^T(\ve{y}-\mu_1 \ve{A}_{02}\ve{1}_{02}) \|_2^2}{(1+\delRIP)\sigma_1^2 + \sigma_e^2} + \frac{\|\underbar{\ve{e}}_{02}\|_2^2}{\sigma_e^2} \label{eqn:Bth2:2nd02ineq2} \\
&\geq& \frac{\| \bar{\ve{W}}_{0 \setminus 2}^T(\ve{y}-\mu_1 \ve{A}_{02}\ve{1}_{02}) \|_2^2 + \| \bar{\ve{W}}_{2 \setminus 0}^T(\ve{y}-\mu_1 \ve{A}_{02}\ve{1}_{02}) \|_2^2}{(1+\delRIP)\sigma_1^2 + \sigma_e^2} + \frac{\|\underbar{\ve{e}}_{02}\|_2^2}{\sigma_e^2} \label{eqn:Bth2:2nd02ineq3} \\
&=& \frac{\| \bar{\ve{W}}_{0 \setminus 2}^T(\ve{A}_0\ve{z}_0 + \bar{\ve{e}}_0) \|_2^2 + \| \bar{\ve{W}}_{2 \setminus 0}^T(\bar{\ve{e}}_2 -\mu_1 \ve{A}_2\ve{1}_2) \|_2^2}{(1+\delRIP)\sigma_1^2 + \sigma_e^2} + \frac{\|\underbar{\ve{e}}_{02}\|_2^2}{\sigma_e^2} \label{eqn:Bth2:2nd02ineq4} \\
&=& \frac{\| \underbar{\ve{U}}_2^T(\ve{A}_0\ve{z}_0 + \bar{\ve{e}}_0) \|_2^2 + \| \underbar{\ve{U}}_0^T(\bar{\ve{e}}_2 -\mu_1 \ve{A}_2\ve{1}_2) \|_2^2}{(1+\delRIP)\sigma_1^2 + \sigma_e^2} + \frac{\|\underbar{\ve{e}}_{02}\|_2^2}{\sigma_e^2} \label{eqn:Bth2:2nd02ineq4b} \\
&\geq& \left(\frac{1-2\delRIP}{1-\delRIP^2} \right)\frac{\| \ve{A}_0\ve{z}_0 + \bar{\ve{e}}_0 \|_2^2 + \| \bar{\ve{e}}_2 -\mu_1 \ve{A}_2\ve{1}_2 \|_2^2}{(1+\delRIP)\sigma_1^2 + \sigma_e^2} + \frac{\|\underbar{\ve{e}}_{02}\|_2^2}{\sigma_e^2}. \label{eqn:Bth2:2nd02ineq5} 
\end{eqnarray}
The last inequality follows from proposition~\ref{prop:bp1}.
Noting that $\|\underbar{\ve{e}}_0 \|_2^2 -\|\underbar{\ve{e}}_{02} \|_2^2 = \|\bar{\ve{e}}_{2 \setminus 0} \|_2^2 \leq \|\bar{\ve{e}}_2 \|_2^2 $,
\begin{eqnarray}
\gamma_2(S_{02}) -\gamma_2(S_{01}) &\geq& \left(\frac{1-2\delRIP}{1-\delRIP^2} \right) \frac{\| \bar{\ve{e}}_2 -\mu_1 \ve{A}_2\ve{1}_2 \|_2^2}{(1+\delRIP)\sigma_1^2 + \sigma_e^2} - \frac{\|\bar{\ve{e}}_2\|_2^2}{\sigma_e^2} \nonumber \\
&& \quad  - \| \ve{A}_0\ve{z}_0 + \bar{\ve{e}}_0 \|_2^2 \left[ \frac{1}{(1-\delRIP) \sigma_1^2 + \sigma_e^2}  - \left(\frac{1-2\delRIP}{1-\delRIP^2} \right) \frac{1}{(1+\delRIP) \sigma_1^2 + \sigma_e^2}\right]\label{eqn:Bth2:2nd02ineq6} 
\end{eqnarray}
Now the last term
\begin{eqnarray}
&&\| \ve{A}_0\ve{z}_0 + \bar{\ve{e}}_0 \|_2^2 \left[ \frac{1}{(1-\delRIP) \sigma_1^2 + \sigma_e^2}  - \left(\frac{1-2\delRIP}{1-\delRIP^2} \right) \frac{1}{(1+\delRIP) \sigma_1^2 + \sigma_e^2}\right] \\
&& \quad < \| \ve{A}_0\ve{z}_0 + \bar{\ve{e}}_0 \|_2^2 \left[ \frac{1}{(1-\delRIP) (\sigma_1^2 + \sigma_e^2)}  - \left(\frac{1-2\delRIP}{1-\delRIP^2} \right) \frac{1}{(1+\delRIP) (\sigma_1^2 + \sigma_e^2)}\right] \\
&& \quad = \left[\frac{\delRIP(4+\delRIP)}{(1-\delRIP)(1+\delRIP)^2} \right] \frac{\| \ve{A}_0\ve{z}_0 + \bar{\ve{e}}_0 \|_2^2}{\sigma_1^2 + \sigma_e^2}  \leq 4  \left[\frac{\delRIP(4+\delRIP)}{(1-\delRIP)(1+\delRIP)^2} \right]  \beta Np < 6 \beta Np\label{eqn:Bth2:2nd02ineq7}
\end{eqnarray}
since $\frac{\delRIP(4+\delRIP)}{(1-\delRIP)(1+\delRIP)^2} < \frac{3}{2}$ for $\delRIP \leq \frac{1}{3}$. 
Also, $\|\bar{\ve{e}}_2\|_2^2/\sigma_e^2 \leq 2 \beta Np $.
Then from~(\ref{eqn:Bth2:2nd02ineq6}) and~(\ref{eqn:Bth2:2nd02ineq7}),
\begin{eqnarray}
\gamma_2(S_{02}) -\gamma_2(S_{01}) &\geq& \left(\frac{1-2\delRIP}{1-\delRIP^2} \right) \frac{\| \bar{\ve{e}}_2 -\mu_1 \ve{A}_2\ve{1}_2 \|_2^2}{(1+\delRIP)\sigma_1^2 + \sigma_e^2} - 8 \beta Np  \label{eqn:Bth2:2nd02ineq8}
\end{eqnarray}
Thus from~(\ref{eqn:BMAP:th2:g1}),~(\ref{eqn:BMAP:th2:g3}), and~(\ref{eqn:Bth2:2nd02ineq8}), 
\begin{eqnarray}
\gamma(S_{02}) - \gamma(S_{01}) &\geq& |S_2|\left[\frac{1}{2}\left(1+\frac{\sigma_1^2}{\sigma_e^2} \left(\frac{1-2\delRIP}{1-\delRIP} \right) \right) + \ln \frac{1-p}{p}  \right] \nonumber \\
&& \quad + \frac{1}{2} \left(\frac{1-2\delRIP}{1-\delRIP^2} \right) \frac{\| \bar{\ve{e}}_2 -\mu_1 \ve{A}_2\ve{1}_2 \|_2^2}{(1+\delRIP)\sigma_1^2 + \sigma_e^2} - 4 \beta Np. \label{eqn:Bth2:2nd02ineq9}
\end{eqnarray}
The coefficient of the term $|S_2|$ is positive for sparse problems when $p \leq \frac{1}{2}$. 
Then for any positive value of $|S_2|$, we reach a contradiction to the assumption that $S_{02}$ is the estimate if
\begin{eqnarray}
 \frac{1}{2} \left(\frac{1-2\delRIP}{1-\delRIP^2} \right) \frac{\| \bar{\ve{e}}_2 -\mu_1 \ve{A}_2\ve{1}_2 \|_2^2}{(1+\delRIP)\sigma_1^2 + \sigma_e^2}  &>& 4\beta Np. \label{eqn:Bth2:2nd02ineq911}
\end{eqnarray}
Since $\delRIP \leq \frac{1}{3}$, it is sufficient to reach a contradiction that,
\begin{eqnarray}
\| \bar{\ve{e}}_2 -\mu_1 \ve{A}_2\ve{1}_2 \|_2^2  &>& 32 \beta Np \sigma_1^2  + 24 \beta Np \sigma_e^2. \label{eqn:Bth2:2nd02ineq912} 
\end{eqnarray}
We note that $32 \beta Np \sigma_1^2  + 24 \beta Np \sigma_e^2 < (4 \sqrt{2 \beta Np} \sigma_1 + 2\sqrt{6 \beta Np} \sigma_e)^2$ and $ \|\bar{\ve{e}}_2 \|_2 \leq \sqrt{2 \beta Np} \sigma_e$. Thus if $|S_2| \geq 1$, and
\begin{eqnarray}
|\mu_1|(1-\delRIP)  &>& \sqrt{2 \beta Np} \sigma_e  + 4 \sqrt{2 \beta Np} \sigma_1 + 2\sqrt{6 \beta Np} \sigma_e \\
&=&  4 \sqrt{2 \beta Np} \sigma_1 + (1+2\sqrt{3})\sqrt{2 \beta Np} \sigma_e \label{eqn:Bth2:2nd02ineq913}
\end{eqnarray}
then $\gamma(S_{02})$ cannot be smaller than or equal to $\gamma(S_{01})$. Thus $S_2$ must be empty. 
Since $\delRIP \leq \frac{1}{3}$, a probabilistic sufficient condition for no false alarm is
\begin{eqnarray}
|\mu_1|  > 6 \sqrt{2 \beta Np} \sigma_1 + 3 \left(\frac{1}{2}+\sqrt{3}\right)\sqrt{2 \beta Np} \sigma_e,
\end{eqnarray}
which holds with probability exceeding $1-3 e^{- Np(\beta-1 - \ln \beta)}$.  \qed

\section{Discussion} \label{sec:BMAP:disc}
From theorem~\ref{th:BMAP1} we see that the energy of the true signal restricted to the missed coefficients 
is of the order of energy in the projection of noise onto the subspace spanned by the true signal. A natural 
question that arises is what can we say about the estimate of the signal $\widehat{\ve{x}}$ obtained by regressing
with the measurement matrix restricted to the columns indexed by $\widehat{S}$~? We mention here that $\widehat{\ve{x}}$
is not an optimal estimate of $\ve{x}$ like MAP or MMSE estimates obtained directly from the observed data. Now 
$\widehat{\ve{x}}$ is given by
\begin{eqnarray}
\widehat{\ve{x}}= \arg \min_{\substack{\ve{x} \in \mathbb{R}^N\\ \ve{x}_{I \setminus \widehat{S}} = \ve{0} }}  \|\ve{y} - \ve{Ax} \|_2^2. \label{eqn:BMAP:regress}
\end{eqnarray}
and it can be easily shown that
\begin{eqnarray}
\widehat{\ve{x}}_{\widehat{S}}&=&\widehat{\ve{x}}_{02}=(\ve{A}^T_{02} \ve{A}_{02})^{-1} \ve{A}^T_{02} \ve{y} = 
\ve{V}_{02} \bar{\ve{\Sigma}}_{02}^{-1}\bar{\ve{U}}_{02}^T (\ve{A}_{0} \ve{x}_{0} + \ve{A}_1 \ve{x}_1+\ve{e}) \\
&=& \ve{V}_{02} \bar{\ve{\Sigma}}_{02}^{-1}\bar{\ve{U}}_{02}^T (\ve{A}_{02} \ve{x}_{02} + \ve{A}_1 \ve{x}_1+\ve{e})
= \ve{x}_{02} + \ve{V}_{02} \bar{\ve{\Sigma}}_{02}^{-1}\bar{\ve{U}}_{02}^T (\ve{A}_1 \ve{x}_1+\ve{e}).
\end{eqnarray}
Now $\|\ve{V}_{02} \bar{\ve{\Sigma}}_{02}^{-1}\bar{\ve{U}}_{02}^T \ve{A}_1 \ve{x}_1 \|_2 \leq \frac{1}{\sqrt{1-\delRIP}}  \frac{\delRIP}{\sqrt{1-\delRIP}}\|\ve{x}_1 \|_2 \leq \frac{\delRIP}{1-\delRIP}K_1 \sqrt{Np} \sigma_e$ and 
$\|\ve{V}_{02} \bar{\ve{\Sigma}}_{02}^{-1}\bar{\ve{U}}_{02}^T \ve{e} \|_2 \leq \sqrt{\frac{\beta}{1-\delRIP}} \sqrt{Np} \sigma_e$. 
Also $\|\ve{x}_1 - \widehat{\ve{x}}_1\|_2 = \| \ve{x}_1 \|_2 \leq K_1 \sqrt{Np} \sigma_e$.
Thus $\|\widehat{\ve{x}} - \ve{x}\|_2 \leq \left(\frac{K_1}{1-\delRIP} + \sqrt{\frac{\beta}{1-\delRIP}}\right) \sqrt{Np} \sigma_e$ with probability exceeding $(1-e^{-Np(2\ln 2- 1)})(1-3 e^{- Np(\beta-1 - \ln \beta)})$. This is optimal in the
sense that even if the true support was known it is not possible to do any better. This also shows that even if there
is any coefficient~$i$ falsely detected, due to the restricted isometry property, it's estimate $\widehat{\ve{x}}_{\{i\}}$ must be small.

\begin{figure}
\begin{center}
\psfrag{k}[l][Bl][1.0]{$K_1$}
\psfrag{b}[l][Bl][1.0]{$\beta$}
\psfrag{probability}[l][Bl][1.0]{Probability}
\includegraphics[width=0.6\columnwidth]{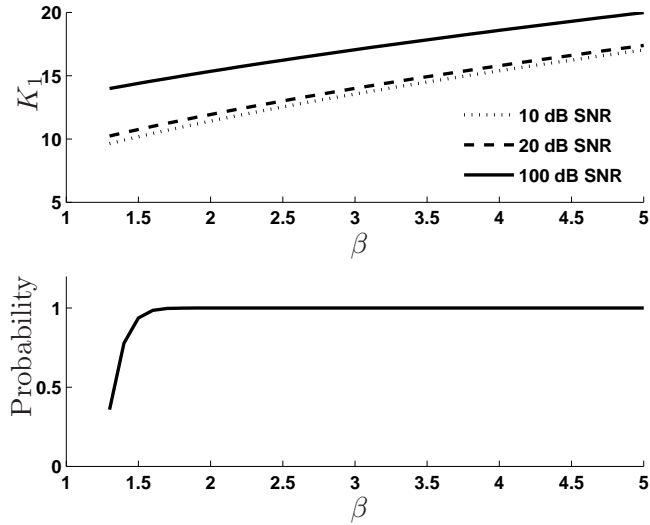}
\caption{The plot in the upper panel shows the constant $K_1$ as a function of the parameter $\beta$. Here
$N=4096, p=0.01, M=256$, $\mu_1=0$ and nominal SNR $10 \log \frac{Np\sigma_1^2}{M\sigma_e^2} =20$~dB.
The figure in the bottom panel shows the least probability with which the energy in the missed 
coefficients is upper bounded by $K_1^2 Np \sigma_e^2$.}
\label{fig:th1}
\end{center}
\end{figure}

Let us analyze the values of the constants appearing in the theorem statements. Consider the example where
$N=4096, p=0.01, M=256$, $\mu_1=0$ and nominal SNR $10 \log \frac{Np\sigma_1^2}{M\sigma_e^2} =20$~dB.
Then for $\beta=1.6$, $K_1=12.94$ and the probability is at least $0.9854$ and for $\beta=2$, 
$K_1=13.77$ and the probability is at least $1-1.06 \times 10^{-5}$. So the constants are modest for reasonable 
values of the system parameters. Fig.~\ref{fig:th1} shows the plots of the constant $K_1$ and the lower bound
of the probability as functions of the parameter~$\beta$ for this example. 
For the same values of $N$, $M$, $p$ and $\frac{\sigma_1^2}{\sigma_e^2}$,
theorem~\ref{th:BMAP2} gives the value of $K_2$ needed to obtain the lower bound on the absolute value of the 
mean $\mu_1$ to probabilistically guarantee perfect support recovery. 
If $\beta=1.6, \bar{\beta}=16$, 
then $K_3=10.75 \sqrt{Np}, K_4=12.94 \sqrt{2Np} $ 
and the probability is at least $0.9832$ and if $\beta=2, \bar{\beta}=25$, 
then $K_3=12.01 \sqrt{Np}, K_4=13.77 \sqrt{2Np} $ 
and the probability is at least $1-4.13 \times 10^{-5}$.

From the statement of theorem~\ref{th:BMAP1} we see that the constant $K_1$ depends on 
$C=\ln(1+\frac{\sigma_1^2}{\sigma_e^2})$.  The term $\frac{\sigma_1^2}{\sigma_e^2}$ is related to SNR. 
We see from Fig.~\ref{fig:th1} that with SNR the constant $K_1$ increases.
So if the SNR increases in an unbounded fashion keeping the noise energy constant then does the energy in the missed 
support grows unbounded? The answer is no. 
If $\sigma_1$ becomes very large then irrespective of the value of $\mu_1$, the probability 
that any element of $\ve{x}$ is close to zero and suppressed by noise becomes very small and 
every element is detected with high probability. 
From~(\ref{eqn:Bth1:ineq4}) we can see that
\begin{eqnarray}
\frac{\left( \sqrt{\frac{1-2\delRIP}{1-\delRIP}} \| \ve{x}_1 \|_2 - \|\bar{\ve{e}}_1\|_2\right)^2}{\sigma_e^2} &\leq& C |S_1| + 8 \beta Np + 4 \beta Np + 2 \beta Np.
\end{eqnarray}
If $|S_1| \neq 0$, the left hand side grows as $\frac{\sigma_1^2}{\sigma_e^2}|S_1|$ whereas the right hand side grows
as $\ln(\frac{\sigma_1^2}{\sigma_e^2})|S_1|$. Thus as SNR grows very large, set $S_1$ has to be empty and there is no
missed coefficient with very high probability. Therefore the upper bound stated in theorem~\ref{th:BMAP1} is loose in 
the very high SNR regime.  For any practical value of the SNR the term $\ln(1+\frac{\sigma_1^2}{\sigma_e^2})$ has a 
moderate value. Hence the constant $K_1$ is a reasonably small constant.

In order to obtain simple expressions in the theorem statements we have used the inequality $\delRIP \leq \frac{1}{3}$ 
instead of having $\delRIP$ appearing in those expressions. As a consequence the constants in the results show the
worst case scenarios when $\delRIP = \frac{1}{3}$. Proceeding in a similar way, for other values of the RIP constant
we can obtain tighter constant values in our results.


\section*{Acknowledgement}
The authors would like to thank Philip Schniter for providing valuable suggestions while this work was in progress 
and carefully reviewing an earlier version of this manuscript.

\bibliographystyle{ieeetr}
\bibliography{BMAP}

\end{document}